\documentclass[pra,amsmath,amssymb,amsfonts,notitlepage]{revtex4-1}

\usepackage{bbm,bm}
\usepackage[final]{hyperref}
\usepackage{amsthm}

\newtheorem{theorem}{Theorem}[section]
\newtheorem{definition}[theorem]{Definition}

\newtheorem{corollary}[theorem]{Corollary}
\newtheorem{lemma}[theorem]{Lemma}

\def\x{\sigma^{\rm x}}
\def\y{\sigma^{\rm y}}
\def\z{\sigma^{\rm z}}
\def\1{\mathbbm{1}}
\def\cL{{\cal L}}
\def\ii{\mathrm{i}}
\def\dot#1#2{{\langle #1 | #2 \rangle}}

\def\tit#1{{#1}, }
\def\etal#1{{#1}}

\begin{document}
\title{\sc \large Translationally invariant conservation laws of local Lindblad equations}

\author{Marko \v Znidari\v c}
\address{
Department of Physics, Faculty of Mathematics and Physics,
  University of Ljubljana, Ljubljana, Slovenia}

\author{Giuliano Benenti}
\address{
CNISM and Center for Nonlinear and Complex Systems,
  Universit\`a degli Studi dell'Insubria, Via Valleggio 11, 22100
  Como, Italy}
\affiliation{Istituto Nazionale di Fisica Nucleare, Sezione di Milano, 
via Celoria 16, 20133 Milano, Italy}

\author{Giulio Casati}
\address{
CNISM and Center for Nonlinear and Complex Systems,
  Universit\`a degli Studi dell'Insubria, Via Valleggio 11, 22100
  Como, Italy}
\affiliation{Istituto Nazionale di Fisica Nucleare, Sezione di Milano, 
via Celoria 16, 20133 Milano, Italy}

\date{\today}

\begin{abstract}
We study the conditions under which one can conserve local translationally invariant operators by local translationally invariant Lindblad equations in one-dimensional rings of spin-$1/2$ particles. We prove that for any $1$-local operator (e.g., particle density) there exist Lindblad dissipators that conserve that operator, while on the other hand we prove that among $2$-local operators (e.g., energy density) only trivial ones of the Ising type can be conserved, while all the other can not be conserved, neither locally nor globally, by any $2$- or $3$-local translationally invariant Lindblad equation. Our statements hold for rings of any finite length larger than some minimal length determined by the locality of Lindblad equation. These results show in particular that conservation of energy density in interacting systems is fundamentally more difficult than conservation of $1$-local quantities. 
\end{abstract}

\maketitle
\tableofcontents

\section{General motivation}
Closed-system evolution dictated by Hamiltonian equations is often an idealization. Systems of interest are typically coupled to external degrees of freedom either on purpose, or because of an inherent unavoidable noise from the environment. Describing the evolution of such systems is in general complicated. Some general conditions should always hold: every quantum evolution has to preserve the trace and positivity of the density matrix. Requiring furthermore that a 
trivially extended evolution is positive also on a larger space leads to the notion of completely positive trace-preserving maps. If such maps depend on a continuous parameter, usually time, and have a semi-group property, meaning that a map for time $t$ can be decomposed into maps for shorter time-steps, then these maps are said to form a dynamical semi-group. It has been shown~\cite{Lindblad,GKS} that every dynamical semi-group is a solution of the Lindblad master equation.

While the description by the Lindblad equation is not the most general one, it is certainly the simplest mathematically consistent master equation generating positive evolution that can be used to describe the dynamics of open quantum systems. Equations of the Lindblad type have in fact been used in physics even before the general formulations of Lindblad and of Gorini, Kossakowski, and Sudarshan
(for instance in laser physics~\cite{Laser} or nuclear magnetic resonance~\cite{NMR}). Despite such rich history, the interest in understanding the properties of the Lindblad equation for systems with many particles is only recently 
beginning to emerge. Indeed, until last couple of years, investigations mainly focused on few-particle systems, like a two-level atom coupled to an electromagnetic field, a system of paramount importance in quantum optics~\cite{Breuer}. Only with recent advances in experimental techniques as well due to new interest coming from condensed and statistical physics, the focus is beginning to
  shift towards many-body systems. Of special interest are the so-called steady (or stationary) states that solve the Liouville equation $\cL(\rho)=0$ and to which the dynamics converges after long time ($\cL$ and $\rho$ are the Liouvillian and the density operator, respectively). It has been for instance shown that dissipative dynamics described by the Lindblad equation can be used to prepare entangled quantum states~\cite{Kraus:08} or to perform universal quantum computation, provided one can control dissipation described by Lindblad operators, driving the system to a steady state where the outcome of the computation is 
encoded~\cite{Verstraete:09}. General properties and conditions for steady states, or more generally for attracting subspaces of Lindblad equations have been studied~\cite{Baumgartner:08,Ticozzi:08}. Especially simple, and thereby well understood, are the so-called dark states -- that is, pure stationary states that are zero eigenstates of each Lindblad operator separately~\cite{Kraus:08} and that can be obtained by local Lindblad equations~\cite{Ticozzi:12}. In the general case of mixed steady states it is known how to construct Lindblad dissipators that lead to a given mixed steady state~\cite{Baumgartner:08} (see also the explicit construction in~\cite{Prosen:09}). Dissipators obtained in such a way are in general non-local. 
On the other hand, a setup with local Lindblad operators,
although rigourously justified only in specific regimes~\cite{Benenti09},
is computationally very 
convenient when investigating thermalization, local equilibrium,  
and transport properties 
of many-body quantum systems. 
Moreover, the setup with local Lindblad operators provides a 
very general paradigm for investigating open many-body quantum systems
with coherent bulk dynamics and incoherent boundary conditions. 
Such approach might find applications in a variety of physical problems,
including the question under what conditions a quantum 
system can be controlled locally, that is, acting on a 
small subsystem only~\cite{Giovannetti}.

In the present work we shall address the question which local translationally invariant operators can be conserved by translationally invariant local Lindblad equations (steady states, in general mixed ones, are special cases of such 
conserved operators). Therefore, we are looking for local translationally invariant conservation laws of local translationally invariant Lindblad equations on finite rings of spin-$1/2$ particles. 

This problem,
besides its fundamental importance, namely, understanding what can and what can not be done with certain classes of Lindblad evolution, has also more 
practical uses. Lindblad equations can be used to study properties of nonequilibrium states, for instance transport far from equilibrium. In such setting 
a desirable tool are dissipators, conserving a given local quantity 
whose transport one wants to study, that though change transport properties. For instance, the dephasing Lindblad operator conserves the 1-body local magnetization (i.e., particle density) and can be used to induce diffusive transport in an otherwise ballistic conductor~\cite{dephasing}. Such a dissipator is very useful in theoretical investigations as one can change transport properties of magnetization from ballistic to diffusive simply by changing the dephasing strength. It would be desirable to have an analogous dissipator that would conserve also other local quantities, for instance local energy density, or even better, both energy density and magnetization. In this way one could independently tune transport properties of energy as well as of magnetization and in doing so obtain for instance a system with high thermomagnetic (or thermoelectric, in 
models with charged particles) efficiency~\cite{ZT}. There is one fundamental difference between the energy and magnetization conservation, which as we shall see, turns out to be very important: for interacting systems energy is usually a 2-body operator while magnetization is a 1-body operator. Therefore, while magnetization densities at different sites commute, in general energy densities at neighboring sites do not. As a consequence, finding energy conserving dissipators is fundamentally more difficult and, as we shall show, local energy-conserving dissipators exist only for the simplest interactions.

Another context in which conservation laws of Lindblad equations have been discussed is in relation to the black hole information paradox. The problem whether information is in fact lost in black holes remains open. In any case, if information is lost, then the theory of quantum gravity can not be unitary. As a consequence, pure states would be allowed to evolve into mixed states~\cite{Hawking}. One simple possibility to describe such evolution, preserving positivity and 
normalization of states, is the Lindblad equation. Here evolution by the Lindblad equation would be an intrinsic one, replacing the Schr\" odinger equation, and not derived after tracing over an external environment. It has been argued that such local Lindblad evolution is incompatible with energy conservation~\cite{Banks,Oppenheim}.
Therefore, conservation of local quantities rigorously studied in lattice systems in the present work is of interest also for very basic considerations in other
fundamental contexts.
  
\section{Formulation of the problem}
We consider the Lindblad equation~\cite{Lindblad,GKS} 
(we set the effective Planck constant $\hbar=1$),
\begin{equation}
\frac{{\rm d}\rho}{{\rm d}t}=\ii [ \rho,H ]+ \cL^{\rm dis}(\rho)=\cL(\rho).
\label{eq:Lin}
\end{equation}
The dissipator can be expressed in a non-diagonal form
\begin{equation}
\cL^{\rm dis}(\rho)=\sum_{j,k} \gamma_{j,k} \left( [L_j\rho,L_k^\dagger]+[L_j,\rho L_k^\dagger]\right),
\label{eq:dis}
\end{equation}
where $L_k$ form an orthogonal operator basis (e.g., for a 2-site $\cL^{\rm dis}$ on spin-$1/2$ particles we have $15$ linearly independent traceless basis operators $\sigma^{\alpha_1} \otimes \sigma^{\alpha_2}$). Hermitian matrix of coefficients $\gamma$ should be non-negative in order to generate a dynamical semi-group. Diagonalizing the structure matrix $\gamma_{j,k}$ we can equivalently write the dissipator in a diagonal form $\cL^{\rm dis}(\rho)=\sum_{j} [L_j\rho,L_j^\dagger]+[L_j,\rho L_j^\dagger]$ (with operators $L_k$ different from those in Eq.~(\ref{eq:dis})). The propagator $\cL$ of the Lindblad equation is called the Liouvillian.

A steady solution of the Lindblad equation is an operator $A$ for which $\cL(A)=0$. We also say that such an operator is conserved by $\cL$ because ${\rm d}A/{\rm d}t=\cL(A)=0$ holds. We shall focus on $r$-local 
(i.e., acting non-trivially on $r$ consecutive sites only)
translationally invariant $\cL$ and study under what conditions is it possible to find a translationally invariant $r$-local Liouvillian $\cL$ (that can contain $r$-local Hamiltonian as well as $r$-local dissipation), so that a given $r$-local Hermitian translationally invariant $A$ is a steady state. That is, writing $A=\sum_j a_j$ and $\cL=\sum_j \cL_j$, where $a_j$ and $\cL_j$ act non-trivially only on $r$ consecutive sites, we want to have
\begin{equation}
\cL(A)=\sum_{j,k} \cL_j(a_k)=0,
\label{eq:SS}
\end{equation}
for a spin-1/2 chain of any length and with periodic boundary conditions. 
Such Lindblad evolution can be said to {\em globally conserve} the $r$-local ``charge'' $a$. We shall specifically focus on 1- or 2-local operators $A$, having in mind conservation of operators like magnetization and energy.

We shall also consider the stronger condition of {\em local conservation}, meaning that
\begin{equation}
\cL_j(a_k)=0,\quad \forall j,k.
\label{eq:sumL}
\end{equation}

We always consider periodic boundary conditions, i.e., rings, so that translational invariance (TI) is exact. An exact conservation (global or local) of some local operator $a_j$ on a ring implies an almost conservation (up-to boundary terms) in a system with open boundaries. Note that there is always a solution with the purely dissipative part $\cL^{\rm dis}$ equal to zero and the Hamiltonian density equal to $a_j$ (or to a function of $a_j$). We are not interested in such trivial solutions; we consider the general, open system case with nonzero dissipative contribution in $\cL_j$.

\section{Summary of results}
We solve the case of 1-local operators by providing conditions under which a given $a_j$ can be globally or locally conserved with an $r$-local Liouvillian $\cL_j$. We also solve the case of conserving more than one linearly independent 1-local $A$. For conservation of 2-local operators we provide a complete picture of local conservation for any $r$-local $\cL_j$, and in the case of global conservation for $2$-local and $3$-local $\cL_j$.

In more precise terms, if the conserved density $a_j$ is a 1-site operator Theorem \ref{th:1Q1Q} states that one can always find a 1-site dissipator that locally conserves that quantity. It is though not possible to find a 1-local Liouvillian that would conserve two linearly independent 1-site densities. Moreover, we prove that it is not possible to locally conserve two linearly independent 1-site operators by any local Liouvillian. However, all $\x,\y,\z$ can be globally conserved by a 2-local translationally invariant Liouvillian (Theorem \ref{th:1Qglob}).    

If $A$ is a 2-local operator Theorem \ref{th:loc2Q} says that local conservation with local Liouvillians is possible if and only if $a_j$ is spanned by $\{\1,u \}\otimes \{\1,w\}$, where $u,w$ are arbitrary 1-site operators. This means that local energy conservation is possible iff the interaction $a_j$ is of the Ising model type (in a longitudinal magnetic field), but is not possible for any other, for instance, for the transverse Ising or the Heisenberg type of $a_j$. 
Theorem \ref{th:2Q2Q} deals with global conservation of a 2-site $a_j$ and $\cL_j$. Although global conservation is at first sight much less restrictive, it turns out that using a 2-site Liouvillian one can conserve only interactions of the type that can be conserved already locally. The same holds also for 3-site Liouvillians.    

\section{Generalities}
We shall consider a one-dimensional lattice of $n$ sites with periodic boundary conditions, i.e., a ring with $n$ sites, each site having two degrees of freedom (a qubit). Any operator $B$ on $n$ sites can be expanded in a product local basis, for instance, taking the basis of Pauli matrices, we have $B=\sum_{\bm{\alpha}} c_{\bm{\alpha}}\, \sigma^{\bm{\alpha}}$, where $\sigma^{\bm{\alpha}}=\prod_{j=1}^n \sigma^{\alpha_j}_j$, and $\bm{\alpha}$ is a vector of length $n$ with each component being from $\alpha_j \in \{0,1,2,3\}$, with the convention $\sigma^0=\x,\,\sigma^1=\y,\,\sigma^2=\z,\,\sigma^3=\1$, while $c_{\bm{\alpha}}$ are expansion coefficients. Hermitian $B$ have real expansion coefficients $c_{\bm{\alpha}}$. A standard inner product used on the space of operators, under which products of local operators form a basis, is the Hilbert-Schmidt inner product $\dot{A}{B}={\rm tr}\,(A^\dagger B)$.

\begin{definition}
An arbitrary product of local operators $\prod_{j=1}^n \sigma^{\alpha_j}_j$ (called a primitive operator) is called {\em $r$-local} iff $r$ is the maximal number of consecutive sites on which two boundary operators are not identity. An operator $B$ is called $r$-local iff it is a sum of $p$-local primitive operators with $p\le r$ (at least one primitive $r$-local term must be nonzero). An operator $B$ is called {\em exactly $r$-local} if it is a sum of only $r$-local primitive operators. An operator $B$ is said to have support on $r$ sites (also shortly that it is an $r$-site operator) iff it acts as an identity on all but $r$ consecutive sites. 
\end{definition}
For instance, $\x_1\1_2\y_3$ or $\x_1\x_2\x_3+\1_1\z_2\x_3$ are 3-local (and have support on $3$ sites), while $\x_1\1_2+\1_1\y_2$ is 1-local (also exactly $1$-local, and has support on $2$ sites). 

We shall also use short notation $b$ for an operator with support on sites $0,\ldots,r-1$, while we denote by $b_j$ the same operator acting on sites $j,\ldots,j+r-1$, that can be obtained from $b$ by a translation, $b_j= T^j(b)$, where $T$ is a translation operator for one site, for instance, $T(\x_1\y_2)=\x_2\y_3$. We state some general lemmas that shall be used in subsequent sections.

\begin{lemma}
\label{lem:sum0}
Let $b$ be an operator that has support on $r$ sites (beware that $b$ is not necessarily $r$-local; it is though a sum of at most $r$-local operators; for instance, $b=\x_0\1_1+\1_0\x_1$ has support on $r=2$ sites, i.e., is a 2-site operator, even-though it is exactly 1-local). A translationally invariant sum $S=\sum_{j=1}^n b_j$ (with periodic boundary conditions) is zero iff $b$ is a linear combination of terms of the form
\begin{equation}
\sum_{k=0}^{r-p} c_k T^k(d),
\label{eq:sum0}
\end{equation}
where $d$ is a $p$-local primitive operator with $1 \le p < r$, and $c_k$ are expansion coefficients that sum to zero, $\sum_{k=0}^{r-p}c_k=0$.
\end{lemma}
\begin{proof} 
If $b$ is of the stated form (\ref{eq:sum0}) we have $S=\sum_{j=1}^n b_j=\sum_{j=1}^n \sum_{k=0}^{r-p} c_k d_{j+k}$. Fixing $j+k$ we see that the coefficient in front of $d_{j+k}$, being equal to $\sum_{k=0}^{r-p}c_k$, is zero. For the other direction of the implication: if we have $S=0$, then $S$ must be orthogonal to any operator. Let us check orthogonality of $S$ to some primitive $p$-local operator denoted by $d$. Because $S$ is a sum of $r$-site operators $b$, a non-zero 
contributions in the overlap $\dot{d_j}{S}$ can come only from terms in which $b_j$ and $d_k$ have non-zero support on the same sites, resulting in the condition $0=\dot{d_j}{S}=\sum_{k=0}^{r-p} \dot{d_k}{b}$. Writing $b$ as a finite sum 
of orthogonal primitive operators and denoting $\dot{d_k}{b}=c_k$, we see $b$ is a linear combination of sums like $\sum_{k=0}^{r-p} c_k d_k$, where $d_k=T^k(d)$ (a linear combination for different $d$'s), while the condition from previous sentence means that $\sum_{k=0}^{r-p}c_k=0$ should hold for any $d$.
\end{proof}
Cases of special importance to us are: (i) a translationally invariant (TI) sum of non-zero 1-site operator $b$ can never be zero (trivially, according to Lemma \ref{lem:sum0}, $p$ should be less than $1$), (ii) a TI sum of an operator $b$ with support on 2 sites ($r=2$) can be zero iff $b$ is of the form $b=\1_0{w}_1-{w}_0\1_1$ (note the notation $w_1=T\,w_0$), where ${w}$ is an arbitrary 1-site operator (for $r=2$ having $p=1$ is the only choice and the two coefficients $c_k$ must be, up-to an overall factor, $\pm 1$); (iii) for $r=3$ we have a possibility $p=1$, in which case $b=c_0 w_0\1_1\1_2+c_1 \1_0 w_1\1_2+c_2 \1_0\1_1w_2$ with $c_0+c_1+c_2=0$, or $p=2$, for which one must have $b=W_{01}\1_2-\1_0W_{12}$.

\begin{lemma}
\label{lem:reduced}
Let us have a general Lindblad equation ${\rm d}\rho/{\rm d}t=\cL(\rho)$ on a bipartite system, with the Hamiltonian $H_{j\mu,k\lambda}$ and the structure matrix of the Lindblad dissipator $\gamma_{j\mu,k\lambda}$ (roman/greek indices refer to the two subsystems). A reduced time-derivative of states of the form $\rho=\1_1 \otimes \sigma$, that is ${\rm tr}_1\,[\cL(\1_1 \otimes \sigma)]$, can be written in terms of a reduced linear map $\cL^{\rm red}$ on the 2nd subsystem that is also of the Lindblad form.
\end{lemma}
\begin{proof}
Taking an orthogonal product basis of Lindblad operators $L_{j\mu}=L'_{j}\otimes L_\mu$, and evaluating trace over the 1st subspace, ${\rm tr}_1({\rm d}\rho/{\rm d}t)$, we get
\begin{equation}
{\rm tr}_1\frac{{\rm d}(\1_1 \otimes \sigma)}{{\rm d}t}=\ii [ \sigma,H^{\rm red}]+ \sum_{\mu,\lambda} \gamma^{\rm red}_{\mu,\lambda}( [L_\mu \sigma,L_\lambda^{\dagger}]+ [L_\mu ,\sigma L_\lambda^{\dagger}]),
\end{equation}
where $H^{\rm red}={\rm tr}_1(H)$ and $\gamma^{\rm red}={\rm tr}_1(\gamma)$. Because partial tracing preserves positivity, $\gamma^{\rm red}$ is also non-negative and therefore the RHS is of the Lindblad form.
\end{proof}

\section{One-local operators}
Here we consider $1$-local Hermitian operators $a_j$ with support on $1$ site. We shall first consider the case when $\cL_j$ is also a $1$-site operator (meaning that each Lindblad operator as well as each term in $H$ has support on a single site).

\subsection{One-site dissipators}
As a side remark we recall that for any single-site operator $w$ it is known~\cite{Prosen:09} how to construct a purely dissipative single-site $\cL_j$, so that $\cL_j(w_j)=0$. For that construction the steady state is nondegenerate, i.e., $w_j$ is the only operator $a_j$ for which $\cL_j(a_j)=0$. Observe though that, due to a TI of the problem we consider, this does not yet guarantee that Eq.~(\ref{eq:SS}) is satisfied. In fact, with such construction Eq.~(\ref{eq:SS}) is generically never satisfied because the terms $\cL_j(\1_j)$ occurring in Eq.~(\ref{eq:SS}) are nonzero as $w$ is the only stationary state.

Going now to our problem, Eq.~(\ref{eq:SS}), and taking into account that $a_j$ and $\cL_j$ are $1$-site operators, we have $\sum_{j\neq k}\cL_j(\1_j)a_k+\sum_k \cL_k(a_k)=0$. Due to trace preservation $\cL_j(\1_j)$ must always be orthogonal to $\1_j$. If $\cL_j(\1_j)$ would be nonzero the part with the sum $\sum_{j\neq k}$ could never sum to zero, because for non-identity $a_k$ the sum of a 2-site $\cL_j(\1_j)a_k+a_j\cL_k(\1_j)$ could never be zero. We therefore conclude that one must have $\cL_j(\1_j)=0$, that is, $\1$ must be a stationary state, i.e., the induced quantum channel must be unital. We shall see that this can be viewed as a special case of a more general condition on unitality given by Lemma \ref{lem:unitalGen}. In addition, the second sum then implies that $a_k$ must also be a stationary state. Demanding global conservation (\ref{eq:SS}) of a $1$-site operator $a_k$ the Liouvillian $\cL_k$ must have at least a doubly degenerate steady state ($\1$ and $a$ must 
 be in the kernel), i.e., $a$ and $\1$ must in fact be locally conserved.

Finding such 1-site Liouvillian is actually easy. Taking a single Lindblad operator $L=\x$ (sometimes called a dephasing) and $H=0$ we see that steady states are spanned by $\{ \x,\1 \}$. Using local unitary rotations we can transform that subspace to any other 1-site operator $a$. What is more, we can see that $a$ is conserved locally (\ref{eq:sumL}), not only globally (\ref{eq:SS}). One can always find a $1$-site $\cL_j$ that locally conserves any $1$-site operator $a$.

One might wonder if it is possible to construct a 1-site Liouvillian for which one would have a steady state subspace of dimension $3$. That is, in addition to $\1$ and $\x$ such $\cL_j$ would also preserve say $\y$ (with a 3 dimensional subspace of steady states $\{\1,w,w'\}$ we can always orthogonalize it and with local unitary rotation bring it to $\{\1,\x,\y\}$. The following Lemma gives a negative answer.
\begin{lemma}
\label{lem:1xy}
A 1-site Liouvillian can have at most a two dimensional subspace of stationary states. In particular, a 1-site Liouvillian that would have a three dimensional kernel (spanned for instance by $\{\1,\x,\y\}$) does not exist.
\end{lemma}
\begin{proof}
Taking $L=\x$ provides an example having a kernel of dimension $2$. If the kernel is of dimension $3$ there is a single state $w$ that is not from the kernel, $\cL(w)\neq 0$. By local unitaries we can bring it to $w=\mu \1+\z$ with some real $\mu$. Kernel is on the other hand spanned by $\{\x,\y,\1-\mu \z\}$. This in particular means that the induced channel $\Lambda={\rm e}^{\cL t}$, for any $t$, would map a state $\1+x\x+y\y-\mu\z$ (with any $x,y$) to itself. In the space of Bloch vectors all states on the crossection of a plane $z=-\mu$ with the Bloch ball would be stationary. This though is not possible for a completely positive map, i.e., a channel, because it would mean that there would be a circle of pure states in the output of $\Lambda$. The so-called no-pancake theorem~\cite{pancake} forbids that; pure outputs of a single qubit quantum channel can namely form either a complete Bloch sphere, a single point, two points, or there is no pure output. 
\end{proof}

\begin{theorem}
\label{th:1Q1Q}
One can always find a 1-site dissipator that locally conserves any 1-site operator. There does not exist a 1-local Liouvillian of the Lindblad form that would globally conserve TI sums of two linearly independent (non-identity) 1-site operators. 
\end{theorem}
\begin{proof}
An explicit solution to the first statement is given by the Lindblad operator $L=\x$ (up-to rotations). Second negative statement is a consequence of the fact that demanding $\sum_{j,k} \cL_j(a_k)=0$ and $\sum_{j,k} \cL_j(b_k)=0$, where $\dot{a_k}{b_k}=0$, forces $\cL_j$ to be unital and also $\cL_j(a_j)=\cL_j(b_j)=0$, which is, by Lemma ~\ref{lem:1xy}, not possible.
\end{proof}

\subsection{Many-site dissipators}
Can one improve by allowing the Liouvillian $\cL_j$ to be an $r$-site operator with $r>1$, and have it conserve more than one $1$-site operator? For local conservation the answer is no.
\begin{lemma}
\label{lem:1Qloc}
It is not possible to locally conserve two (or more) linearly independent non-identity TI 1-local operators with a TI $r$-local Lindblad equation (for any $r \ge 1$).
\end{lemma}
\begin{proof}
Due to local conservation we immediately see that $\cL_j$ has to be unital, $\cL_j(\1_j\cdots \1_{j+r-1})=0$. Suppose that we demand local conservation of a 1-site $a_j=\x_j+\mu \1_j$ (using local unitaries we can bring an arbitrary $w$ to such a form). Due to unitality $\cL_j$ should therefore also conserve $\x_j$ alone. A similar argument for the second independent conserved operator brings us to the conclusion that $\cL_j$ should conserve also say $\y_j$. In addition to unitality, we therefore have $\cL_j(\1_j\cdots \x_{j+r-1})=0$ and $\cL_j(\1_j\cdots \y_{j+r-1})=0$. Tracing over all but the last site we see that the reduced 1-site Liouvillian $\cL_j^{\rm red}$ (that is also of the Lindblad type due to Lemma~\ref{lem:reduced}) would have to have $\1,\x,\y$ in the kernel. This though is not possible due to Lemma \ref{lem:1xy}.
\end{proof}
This Lemma, together with Theorem \ref{th:1Q1Q}, means that with local conservation, at least as far as the dimension of the steady subspace is concerned, using larger locality than 1-local Liouvillians brings no advantage. With global conservation though things are different. It turns out that with global conservation one can conserve all 1-site operators already with a 2-local TI Liouvillian.
\begin{theorem}
\label{th:1Qglob}
With local conservation one can not conserve two (non-identity) linearly independent 1-site operators. With global conservation and 2-local $\cL_j$ one can conserve all four linearly independent 1-site operators.
\end{theorem}
\begin{proof}
First part of Theorem is proved in Lemma \ref{lem:1Qloc}. Second part is shown by an explicit construction. Let us take a TI 2-local dissipator with a single Lindblad operator $L=\x_0\x_1+\y_0\y_1+\z_0\z_1$. It is easy to check that the steady states of such 2-site $\cL_0$ are spanned by all symmetric operators (10 in number)
\begin{equation}
\{\x_0\x_1,\x_0\y_1+\y_0\x_1,\y_0\y_1,\x_0\z_1+\z_0\x_1,\y_0\z_1+\z_0\y_1,\z_0\z_1,\x_0\1_1+\1_0\x_1,\y_0\1_1+\1_0\y_1,\z_0\1_1+\1_0\z_1,\1_0\1_1 \}. 
\label{eq:Heis}
\end{equation}
In addition, one has $\cL_0(\x_0\1_1)=-8(\x_0\1_1-\1_0\x_1)$ and similarly for $\y$ and $\z$. Due to (\ref{eq:Heis}) one also has $\cL_j(\x_j\1_{j+1}+\1_j\x_{j+1})=0$, and similarly for $\y$ and $\z$. Also, $\cL_{j-1}(\x_j\1_{j+1}+\1_j\x_{j+1})=\cL_{j-1}(\1_{j-1}\x_j)=-8(\1_{j-1}\x_j-\x_{j-1}\1_j)$, and as a consequence, taking $a=\x$ the sum in Eq.~(\ref{eq:SS}) is zero. The same holds also for $\y_j\1_{j+1}+\1_j\y_{j+1}$ and $\y_j\1_{j+1}+\1_j\y_{j+1}$. Such $\cL_j$ therefore globally conserves all 4 primitive 1-site operators $\x,\y,\z$ and $\1$.
\end{proof}
For 1-local TI operators global conservation with a 2-local Liouvillian is therefore more powerful than local conservation with an arbitrary TI $r$-local $\cL_j$.

\section{Two-local operators - local conservation}
Here we demand local conservation given by Eq.~(\ref{eq:sumL}), which is obviously stronger condition than global conservation. We also consider $2$-site operators $a_k$ resulting in a 2-local $A=\sum_j a_j$. Doing an operator Schmidt decomposition of $a$ we can write it as
\begin{equation} 
a=\sum_{j=1}^r u^j_0\, w^j_1,
\label{eq:schmidt}
\end{equation}
where $u^j_0$ are orthogonal 1-site operators on site 0, while $w^j_1$ are orthogonal on site 1. The rank $r$ is at most 4. Note that, due to translational invariance of $A$, we can always choose all $1$-local terms in $a$ to be symmetric with respect to two sites.

\subsubsection{Ising-like operator}
Using local unitaries any rank $1$ operator $a_j$ can be brought, up-to an irrelevant prefactor, to a form $a_j=(\x_j+\mu\1_j) (\x_{j+1}+\mu'\1_{j+1})$ (choosing symmetric 1-local terms we can in fact take $\mu=\mu'$). Taking a single 2-site Lindblad operator $L=\x_0\x_1$ and no unitary part, $H=0$, one can see that the stationary states of the corresponding 2-site dissipator $\cL_0$ are spanned by operators $\{\x_0\x_1,\x_0\1_1,\1_0\x_1,\y_0\y_1,\z_0\z_1,\y_0\z_1,\z_0\y_1,\1_0\1_1 \}$. Eq.(\ref{eq:sumL}) is therefore satisfied and $L=\x_0\x_1$ locally conserves any 2-site $a_j$ spanned by $\{\x_0\x_1,\x_0\1_1,\1_0\x_1,\1_0\1_1\}$, i.e., any Ising nearest-neighbor coupling in an optional longitudinal field. Such $a_j$ are in fact not necessarily of rank $1$, they can also be of rank $2$, like for instance $\x_0\x_1+\1_0\1_1$. Observe that other basis states, like e.g. $\y_0\z_1$, are not locally conserved by $L=\x_0\x_1$ because $\cL_{j-1}(\y_j \z_{j+1})=\cL_{j-1}(\1_{j-1}\y_j
 )\z_{j+1}\neq 0$.

\subsubsection{Other operators}
Let us first discuss $2$-site Liouvillians. We consider $a_j$ that are not from ${\rm span}\{\x_0\x_1,\x_0\1_1,\1_0\x_1,\1_0\1_1\}$. They have the Schmidt decomposition of the form $a_j=(\x_j+\mu\1_j)(\x_{j+1}+\mu'\1_{j+1})+\sum_{k=2}^4 u^k_j w^k_{j+1}$. By local unitary rotation we can furthermore rotate $u^2_j$ to $u^2_j=\mu \x_j-\1_j+\lambda \y_j$ with some real $\lambda$. For $a_k$ that is not Ising-like $\lambda$ is always nonzero (to see that, we choose among $u^{2,3,4}_j$ the one not from ${\rm span}\{\x_j,\1_j\}$, and rotate it around the $x$-axis). Looking at local conservation $\cL_j(a_k)=0$, such that sites $(j,j+1)$ do not overlap with sites $(k,k+1)$, we conclude that one must have $\cL_j(\1_j\1_{j+1})=0$. For local conservation the dissipator has to be unital (this in fact holds for any $r$-site $\cL_j$ and any $p$-site $a_j$). Taking $j=k-1$ in Eq. (\ref{eq:sumL}) we get $\sum_{k=1}^4 \cL_j(\1_j u^k_{j+1})w^k_{j+2}=0$. Because operators $w^k_{j+2}$ are orthogonal, each term in the sum has to be separately zero, 
$\cL_j(\1_j u^k_{j+1})=0$. Taking $u^1_{j+1}=(\x_{j+1}+\mu\1_{j+1})$ and using unitality we get that $\cL_j(\1_j\x_{j+1})=0$. Demanding also that $\cL_j(\1_j (\mu \x_{j+1}-\1_{j+1}+\lambda \y_{j+1}))=0$, we conclude that for $\lambda \neq 0$ we must in addition have $\cL_j(\1_j\y_{j+1})=0$. Tracing over site $j$ we get the reduced Liouvillian for which $\cL^{\rm red}(\x)=\cL^{\rm red}(\y)=\cL^{\rm red}(\1)=0$. Because the reduced dissipator $\cL^{\rm red}$, giving a mapping of states of form $\1_j w_{j+1}$, again has to be of the Lindbladian form (Lemma~\ref{lem:reduced}), and such three times degenerate stationary state is according to Lemma~\ref{lem:1xy} not possible, we conclude that Eq.(\ref{eq:sumL}) can not be fulfilled for a 2-site operator $a$, unless $\lambda=0$ (in which case $a$ is linear combination of $\{ \x_j\x_{j+1},\1_j\x_{j+1},\x_j\1_{j+1},\1_j\1_{j+1}\}$, up-to local unitaries, for which Lindblad operator $
 L=\x_j\x_{j+1}$ works).

We can actually see that the above argument readily generalizes to TI $r$-local Lindblad operators for any $r>2$. For an infinite system size one first concludes that $\cL_j$ must be unital, and then, similarly as above, that $\cL_j(\1_j\cdots\1_{j+p-2}\x_{j+p-1})=\cL_j(\1_j\cdots\1_{j+p-2}\y_{j+p-1})=0$. The reduced 1-qubit channel would therefore have to have a triply degenerate stationary state, which is not possible. No matter on how many sites the Liouvillian $\cL_j$ acts one can never have $\cL_j(a_k)=0$ for a 2-site $a_k$ that is not from $\{\x,\1\}^{\otimes 2}$. The above finding can be summarized in the following theorem.
\begin{theorem}
\label{th:loc2Q}
A 2-site operator $a_j$ can be locally conserved by an $r$-site Liouvillian $\cL_j$ (finite $r\ge 2$), $\cL_j(a_k)=0,\, \forall j,k$, iff $a_j$ is of the Ising type, that is, if up-to local unitary rotations $a_j$ is linear combination of operators $\x_j\x_{j+1}$, $\1_j\x_{j+1}$, $\x_j\1_{j+1}$, and $\1_j\1_{j+1}$.
\end{theorem}

\section{Two-local operators -- global conservation}
Theorem \ref{th:loc2Q} is a bit disappointing as it means that any non-trivial energy density, e.g., that of the transverse Ising model, the Heisenberg model, etc., can not be locally conserved by any local TI Lindblad equation. This must be contrasted with the case of 1-site operators that can always be locally conserved already by a 1-site dissipator (Theorem \ref{th:1Q1Q}). One might think that replacing the local conservation (\ref{eq:sumL}) by the global conservation (\ref{eq:SS}) will increase the set of $2$-site $a_k$ that can be conserved. Somewhat unexpectedly though, as we shall show in this Section, this is not the case. Even under global conservation $\sum_{j,k}\cL_j(a_k)=0$ a solution for 
2- and 3-site $\cL_j$ exists only for the Ising-like $a_k$.

\subsection{Constraint on unitality} 
We shall first derive a constraint on unitality, i.e., on the image of the identity, that the global conservation puts on local $\cL_j$, in particular on a $2$-local $\cL_j$. 
\begin{lemma}
\label{lem:unital}
Having global conservation $\cL(A)=0$, with a 2-local TI $A=\sum_j a_j$, and a 2-local TI $\cL=\sum_j \cL_j$, for a system of any length, imposes a divergence-like condition
\begin{equation}
\cL_j(\1_j\1_{j+1})=w_j\1_{j+1}-\1_j w_{j+1},\qquad \dot{w}{\1}=0,
\label{eq:unital}
\end{equation}
where $w$ is some $1$-site operator.
\end{lemma}
\begin{proof}
Suppose we have $\cL_0(\1_0\1_1)=W$ with some nonzero 2-site operator $W$. Projecting the l.h.s. in global conservation Eq.(\ref{eq:SS}) on an arbitrary operator, i.e., calculating $\langle B | \cL(A) \rangle$, we must get $0$. Let us project on two primitive 2-local operators $B_k=\sigma_k^{\alpha_k} \sigma_{k+1}^{\alpha_{k+1}}$ and $R_j=\sigma_j^{\alpha_j} \sigma_{j+1}^{\alpha_{j+1}}$ such that their support sites do not overlap, $|k-j|>2$ (all $\alpha_k,\alpha_{k+1},\alpha_j,\alpha_{j+1}\neq 3$). Projecting Eq.(\ref{eq:SS}) on $B_k R_j$, there are only two terms among $\cL_p(a_r)$ that can possibly have a non-identity operator on all four sites. Namely, either $p=k$ and $r=j$, or, $p=j$ and $r=k$. Projection therefore gives a condition $\dot{B_k}{a_k}\dot{R_j}{W_j}+\dot{B_k}{W_k}\dot{R_j}{a_j}=0$, that holds for all allowed sets of four $\alpha \neq 3$. Because $a_j$ is 2-local (and 2-site) operator there must be at least one primitive 2-local operator $U$ (a direct product of two non-identity operators) that is not orthogonal 
to $a_j$, $\dot{U_j}{a_j}\neq 0$. Choosing $B=R=U$ we immediately see that $W$ is orthogonal to $U$. Taking any other primitive 2-local operator $\sigma_0^{\alpha_0}\sigma_1^{\alpha_1}$ there are two possibilities: i) such term is orthogonal to $a_j$. In this case we take $B=U$ and $R=\sigma_0^{\alpha_0}\sigma_1^{\alpha_1}$, concluding that $\dot{R}{W}=0$; ii) it is not orthogonal to $a_j$, in which case we take $B=R=\sigma_0^{\alpha_0}\sigma_1^{\alpha_1}$, concluding again that $\dot{R}{W}=0$. $W$ is therefore orthogonal to all primitive 2-local operators and must be of the form $W=\1_0 w_1+w_0'\1_1$. Trace preservation also imposes that $w$ and $w'$ are both orthogonal to $\1$. Projecting now Eq.~(\ref{eq:SS}) on a 1-site $B_k=u_k\1_{k+1}$ and $R_j=U_j$, we get $\dot{U_j}{a_j}(\dot{B_k}{W_{k-1}}+\dot{B_k}{W_k})=0$, giving in turn $\dot{u_k}{w_k}+\dot{u_k}{w_k'}=0$ for any non-identity $u_k$. We conclude that $w=-w'$. 
\end{proof}
Lemma \ref{lem:unital} is very useful because it enables us to evaluate one sum in a double sum in the global conservation Eq.(\ref{eq:SS}), reducing it to a single sum,
\begin{equation}
\sum_j \cL_j(a_j)+\cL_{j-1}(a_j)+\cL_{j+1}(a_j)+a_j w_{j+2}-w_{j-1} a_j=0.
\label{eq:sum00}
\end{equation}
From now on this equation will serve us as a staring point for global conservation. Another consequence of Lemma \ref{lem:unital} is the following corollary.
\begin{corollary}
\label{cor:11}
Let $A$ be a TI 2-local operator, $A=\sum_j a_j$, that is globally conserved under TI 2-local Liouvillian $\cL$, $\cL(\sum_j a_j)=0$. Then a shifted operator $\tilde{a}_j=a_j+\mu \1_j\1_{j+1}$, with an arbitrary $\mu$, is also conserved under the same $\cL$, $\cL(\sum_j \tilde{a}_j)=0$. As a consequence, it is enough to consider global conservation of $2$-site operators $a_j$ that are orthogonal to the identity. 
\end{corollary}
\begin{proof}
The Corollary is a simple consequence of the linearity of $\cL$ and of Lemma \ref{lem:unital}. Namely, we have $\cL(\mu\1)=\mu\sum_j \cL_j(\1_j\1_{j+1})=\sum_j w_j\1-\1 w_{j+1}=0$. In other words, the Liouvillian $\cL=\sum_j \cL_j$ is unital.
\end{proof}

From the way the proof of Lemma ~\ref{lem:unital} proceeds one can see that the Lemma can be readily generalized beyond $2$-local operators.
\begin{lemma}
\label{lem:unitalGen}
Demanding global conservation $\cL(A)=0$ for an $m$-local TI $A=\sum_j a_j$ and an $n$-local TI $\cL=\sum_j \cL_j$ for a system of sufficient length ($m$ and $n$ are fixed), enforces $\cL_j$ to map the identity to a linear combination of terms of the type in Eq.~(\ref{eq:sum0}) with $p < n$.

Using Lemma \ref{lem:sum0} this in fact means that global conservation of a local TI operator imposes that $\cL=\sum_j \cL_j$ must be unital!
\end{lemma}
\begin{proof}
Case $m=n$. Let us denote $\cL_0(\mathbbm{1}_0\cdots \mathbbm{1}_{n-1})=W$. Because $a_j$ is $m$-local it is non-orthogonal to at least one primitive $m$-local $U$. Similarly as in the proof of Lemma \ref{lem:unital}, taking $B_k=R_j=U$, we conclude that $\dot{U}{W}=0$ (an index $j$ in for instance $B_j$ indicates the smallest site index on which $B$ acts nontrivially). Then, take $B_k=U_k$ and for $R_j$ any primitive $p$-local operator ($p \le n$). Projecting global conservation on these two operators and noting that for $p<n$ $R_j$ has non-zero overlap with $W$ on different sites, we get $\dot{U_k}{a_k}\dot{R_j}{\sum_{r=0}^{n-p} T^{p-n+r} W_j}=0$, which means that $\sum_{r=0}^{n-p} \dot{R_{j+r}}{W_j}=0$. A sum of expansion coefficients of $W$ on any primitive $p$-local $R$ must be zero, which is nothing but the condition in Lemma \ref{lem:sum0}.

Case $n>m$. Take a primitive $m$-local $U$ and $B_k=U_k$, such that $\dot{U}{a}\neq 0$, and any primitive $p$-local $R_j$, with $n\ge p>m$. Projecting and noting that, because $p>m$, $B_k$ must be projected on $a$ while $R_j$ on $W$, we see that $\sum_{r=0}^{n-p} \dot{R_{j+r}}{W_j}=0$. To check the expansion of $W$ on $(p=m)$-local primitive operators we first take $R_j=U_j$, resulting in $\sum_{r=0}^{n-m} \dot{U_{j+r}}{W_j}=0$. Using this we can then see that, taking for $R_j$ any other $p$-local primitive operator, with $p\le m$, and $B_k=U_k$, results again in $\sum_{r=0}^{n-p} \dot{R_{j+r}}{W_j}=0$. The sum of expansion coefficients therefore sums to zero for any primitive operator with $p\le n$.

Case $n<m$. Take a primitive $m$-local $U$ and $B_k=U_k$, such that $\dot{U}{a}\neq 0$. Because of $m>n$ $U$ is orthogonal to $W$ and therefore, taking for $R_j$ any primitive $p$-local operator, where $p\le n$, we immediately obtain $\sum_{r=0}^{n-p} \dot{R_{j+r}}{W_j}=0$.
\end{proof}

As a consequence of Lemma, for instance, for $n=3$ one has to have 
\begin{equation}
\cL_j(\1_j\1_{j+1}\1_{j+2})=c_0 w_j\1_{j+1}\1_{j+2}+c_1\1_j w_{j+1}\1_{j+2}+c_2\1_j\1_{j+1}w_{j+2}+\sum_P c_P(P_{j,j+1}\1_{j+2}-\1_j P_{j+1,j+2}),
\label{eq:unitality3}
\end{equation}
with $c_0+c_1+c_2=0$ and the 2nd sum being over all primitive 2-local operators $P$.

\subsection{Two-site dissipators}
Here we are going to consider global conservation with $2$-local Lindblad equations. Let us first derive a useful form to which any $2$-site operator $a_j$ can be brought to by local unitary rotations (that is by $U\otimes U'$). Due to Corollary \ref{cor:11} it is enough to consider operators that are orthogonal to $\1_j\1_{j+1}$. Any such operator can be split into a purely 2-site part $\tilde{a}$ (that is a sum of primitive 2-local operators) and a 1-site 
``magnetic'' field-like term $w$, $a=\tilde{a}+w_0\1_1+\1_0 w'_1$. Making an operator Schmidt decomposition on $\tilde{a}$, we can write it as $\tilde{a}=\sum_{j=1}^3 u_0^j v_1^j$. Rotating three orthogonal operators $u^j$ and $v^j$ with local unitaries to $\x,\y$ and $\z$, and noting that for our TI operator $A$ we can always symmetrize $w_0\1_1+\1_0 w'_1$ and write it as $\frac{1}{2}((w_0+w'_0)\1_1+\1_0 (w_1+w'_1))$, we get a ``canonical'' form of 2-site operators,
\begin{equation}
a=\x_0\x_1+\mu \y_0\y_1+\nu \z_0\z_1+w_0\1_1+\1_0w_1,\quad w=h_x \x+h_y \y+h_z\z.
\label{eq:canonical}
\end{equation}
Without sacrificing generality we can assume that $\mu,\nu \in [0,1]$ (we rotate $u^j$ and $v^j$ with the largest weight to $\x$; if $\mu<0$ and $\nu>0$ we can change the sign of $\x$ and $\z$ by local unitaries; if $\mu,\nu<0$ we change the sign of $\y,\z$, bringing both $\mu$ and $\nu$ to non-negative values, apart from an irrelevant overall factor). Magnetic field strengths $h_{x,y,z}$ are arbitrary real numbers.

We shall now show that, apart from the simplest Ising-like case for which $a_j=\x_j \x_{j+1}$ (or more generally, $a_j$ is from $\{\x,\1\}^{\otimes 2}$) it is not possible to find a TI 2-site Liouvillian fulfilling global conservation condition (\ref{eq:sum00}). Relaxing conservation from local to global we therefore do not gain anything -- with global conservation only those cases already solvable by local conservation are possible. This is summarized in the following theorem.
\begin{theorem}
\label{th:2Q2Q}
A generic 2-site interaction $a$ of the form 
\begin{equation}
\label{eq:2sitea}
a=\x_0\x_1+\mu\, \y_0\y_1+\nu\, \z_0\z_1+w_0\1_1+\1_0w_1,\quad w=h_x \x+h_y \y+h_z\z,
\end{equation}
can be globally conserved by a 2-site Lindblad superoperator $\cL_j$, $\sum_j \cL_j(A)=0,\,\, A=\sum_j a_j$, iff $\mu=\nu=h_y=h_z=0$.
\end{theorem}
\begin{proof}
If $a$ is of the form stated in Theorem ($\mu=\nu=h_y=h_z=0$) we have already showed in Theorem \ref{th:loc2Q} that such $a$ can be locally conserved. Proof of the other direction is more involved. General idea is to show that Eq.~(\ref{eq:sum00}) is incompatible with complete-positivity of Lindbladian evolution under given constraints. 

The proof proceeds in three steps: (i) We consider a Liouvillian $\cL_j$ with a single general Lindblad operator $L=\sum_{\bm{j}} c_{\bm{j}} \sigma_0^{j_1}\sigma_1^{j_2}$ and an arbitrary unitary part (i.e., Hamiltonian). Using Lemma~\ref{lem:sum0} we write out equations obtained by projecting Eq.~(\ref{eq:sum00}) to various 1-, 2-, or 3-site operators; (ii) We sum certain equations together in order to get rid of the (linear) dependence on the unitary part of $\cL_j$, thereby obtaining a quadratic form in unknown expansion coefficients $c_j$ of $L$, that equates to zero, $c_i^* c_j C_{ij}=0$. It turns out that, by properly choosing equations that we sum, we can achieve that the matrix $C$ is negative-definite and therefore the only solution is a trivial one $c_j=0$, i.e., no 2-site Lindblad 
operator exists for which one could solve a certain combination of equations obtained from Eq.(\ref{eq:sum00}). (iii) If we have more than one Lindblad operator we get a sum of negative-definite quadratic forms, one for each Lindblad operator, again concluding that a solution does not exist. Crucial steps are (i) and (ii), finding the right combination of equations to get a negative-definite $C$ (here a symmetry of the canonical form of $a$ (\ref{eq:canonical}) is helpful) and then proving that $C$ is actually negative-definite, except for special values of parameters. Details of the all three steps of the proof can be found in Appendix \ref{app1}. 
\end{proof}

\subsection{Non two-local operators with support on two sites}

Crucial reason why it is not possible to have a $2$-local Lindbladian that would conserve a non-trivial $2$-local operator is the connectivity of the one-dimensional lattice as reflected in the TI sum $A=\sum_j a_j$ in which $a_j$ acts only on nearest-neighbor sites $j$ and $j+1$. If one relaxes this condition then it is possible to find a 2-local Liouvillian that conserves such $A$.

As a simple example, let us define $a_{j,k}$ as an operator $a$ that acts nontrivially only on $j$-th and $k$-th sites and let $A=\sum_{j,k} a_{j,k}$. Such $A$ can in turn be conserved by a 2-local dissipator. Namely, taking a single Lindblad operator $L=\x_0\x_1+\y_0\y_1+\z_0\z_1$ and no Hamiltonian evolution, one can calculate that the stationary states of such 2-site dissipator are spanned by 
10 basis states $\{\x_0\x_1,\x_0\y_1+\y_0\x_1,\y_0\y_1,\x_0\z_1+\z_0\x_1,\y_0\z_1+\z_0\y_1,\z_0\z_1,\x_0\1_1+\1_0\x_1,\y_0\1_1+\1_0\y_1,\z_0\1_1+\1_0\z_1,\1_0\1_1 \}$. In addition, one has $\cL(\x_0\1_1)=-8(\x_0\1_1-\1_0\x_1)$ and similarly for $\y$ and $\z$. This then means that if we take for 2-local operator $a$ any permutationally invariant operator (any such state is spanned by the above basis of steady states) the corresponding $A=\sum_{j,k} a_{j,k}$ (which is not 2-local; it is though a sum of operators having support on 2 sites) will be conserved by such non-nearest-neighbor 2-site permutationaly invariant $\cL=\sum_j \cL_j$.

\subsection{Beyond two-site dissipators}
Considering that for a 1-site $a_j$ using more than 1-site Liouvillian enabled one to do more than just with 1-site Liouvillians (see Theorem (\ref{th:1Qglob}) one can wonder whether more than 2-site Liouvillians could perhaps conserve non-trivial 2-site $a_j$. We do not know a complete answer to this question; here we present result for the case that $\cL_j$ is a 3-site operator.

We can treat a 3-site $\cL_j$ using the same idea used in the 2-site case, see Theorem \ref{th:2Q2Q} and its proof. One constructs a negative matrix $C$ that is negative-definite, except at special values of parameters of a 2-site $a_j=\x_j\x_{j+1}+\mu \y_j \y_{j+1}+\nu \z_j\z_{j+1}+\1_j(h_x\x_{j+1}+h_y \y_{j+1}+h_z \z_{j+1})$. For a Liouvillian with a single Lindblad operator, now depending on $4^3-1$ complex coefficients $c_j$, global conservation of $a_j$ imposes a condition $\bm{c}^\dagger C \bm{c}=0$. We shall not write out the $63\times 63$ matrix $C$ that depends on $\mu,\nu,h_{x,y,z}$, but just state the result. It turns out that, exactly as in the 2-site case, $C$ is negative-definite except in the trivial case of $\mu=\nu=h_y=h_z=0$. Details of the proof, that is very similar to the one for a $2$-local case, can be found in Appendix~\ref{AppB}.

\section{Conclusion}
We have studied conservation laws of local translationally invariant Lindblad equations on a ring. We have solved the problem of conservation of translationally invariant $1$-local operators, showing that any such operator can be locally conserved already by a $1$-local Lindblad equation. Considering simultaneous conservation of more than one $1$-local operator we have also proved that all $1$-local operators can be globally conserved by $2$-local Lindblad equations. For $2$-local translationally invariant operators the results are quite different. We have proved that one can locally conserve all Ising-like $2$-local operators with a $2$-local Lindblad equation. On the other hand, all other $2$-local operators can not be conserved, neither locally nor globally, by any $2$-local and not even by $3$-local Lindblad equation. The problem for $p$-local Lindblad equations, with $p>3$, remains open. 

Our rigorous results in particular show that the conservation of any nontrivial interaction energy (being a $2$-local operator) is fundamentally different than conservation of $1$-local operators -- in fact, it can not be done by any sufficiently local translationally invariant Lindblad equation. Therefore, in order to conserve the energy one is forced to either relax the locality constraint, allowing for non-local Lindbladians, or, to relax translational invariance, for instance, by allowing for locally-tuned reservoirs as is often done in mesoscopic physics for electric~\cite{LB}, thermal~\cite{thermal},
and thermoelectric transport~\cite{3ter}. 

\appendix

\section{Proof of Theorem \ref{th:2Q2Q}}
\label{app1}
We shall first construct the step (i) from short description of the proof of Theorem \ref{th:2Q2Q}. We demand that the TI operator obtained from $a_j$ is a stationary state, meaning that operator Eq.~(\ref{eq:sum00}) should be satisfied, and we consider dissipator with a single 2-site Lindblad operator $L=\sum_{\bm{j}} c_{\bm{j}} \sigma_0^{j_1}\sigma_1^{j_2}$with $15$ unknown complex coefficients $c_j$ and a general 2-site Hamiltonian $h=\sum_{j_1,j_2} d_{\bm{j}} \sigma_0^{j_1}\sigma_1^{j_2}$ with $16$ unknown real coefficients $d_{\bm{j}}$ ($H=\sum_j h_j$). We can project the Eq.~(\ref{eq:sum00}) on any operator, thereby obtaining algebraic equations that the coefficients $c_{\bm{j}}$ and $d_{\bm{j}}$ should satisfy. Due to a 2-local nature of $a_j$ and $\cL_j$ and terms in Eq.~(\ref{eq:sum00}) being at most 3-local all nontrivial equations are obtained by 
projecting on all primitive 1-, 2- and 3-local operators. We shall in fact need to project only on three 2-local operators, $\x_0\x_1$, $\y_0\y_1$, $\z_0\z_1$, and on three 1-local operators $\x_0,\, \y_0$, and $\z_0$. Here and in the following we use a short notation $\cL_{mn,kl}$ for the matrix elements of a 2-local Liouvillian $\cL_j$, $\cL_{mn,kl}\equiv [\cL_j]_{mn,kl}=\langle \sigma_j^m \sigma_{j+1}^n | \cL_j(\sigma_j^k \sigma_{j+1}^l) \rangle$, for instance, $\cL_{xy,x\1}=\langle \sigma_j^{\rm x} \sigma_{j+1}^{\rm y} | \cL_j(\sigma_j^{\rm x} \1_{j+1}) \rangle$. Let us first write the equation obtained by projecting Eq.~(\ref{eq:sum00}) on $\x_0\x_1$. Because $\x_0\x_1$ is a 2-local operator and the TI sum in Eq.~(\ref{eq:sum00}) involves at most a 3-local terms, according to Lemma~\ref{lem:sum0} the coefficients in front of $\x_j\x_{j+1}$ should sum to zero, resulting in the equation,
\begin{eqnarray}
\cL_{xx,xx}&+&\cL_{\1 x,\1 x}+\cL_{x\1,x\1}+\mu\cL_{xx,yy}+\nu\cL_{xx,zz}+2h_x(\cL_{xx,\1 x}+\cL_{xx,x\1})+\nonumber \\
&+&2h_y(\cL_{xx,\1 y}+\cL_{xx,y\1})+2h_z(\cL_{xx,\1 z}+\cL_{xx,z\1})+h_x(\cL_{\1 x,\1\1}+\cL_{x\1,\1\1}) =0.
\label{eq:xx}
\end{eqnarray}
Note that we used the condition on unitality, stating that $\cL_j(\1_j\1_{j+1})=w'_j\1-\1 w'_{j+1}$ holds, where $w' \perp \1$. The last term in the above equation is in fact nonzero only if $\cL_j$ is non-unital. The equations obtained by projection on $\y_0\y_1$ and $\z_0\z_1$ are similar; they can be obtained from (\ref{eq:xx}) by appropriately permuting indices. For $\y_0\y_1$ we get
\begin{eqnarray}
\mu(\cL_{yy,yy}&+&\cL_{\1 y,\1 y}+\cL_{y\1,y\1})+\cL_{yy,xx}+\nu\cL_{yy,zz}+2h_x(\cL_{yy,\1 x}+\cL_{yy,x\1})+\nonumber \\
&+&2h_y(\cL_{yy,\1 y}+\cL_{yy,y\1})+2h_z(\cL_{yy,\1 z}+\cL_{yy,z\1})+h_y(\cL_{\1 y,\1\1}+\cL_{y\1,\1\1}) =0,
\label{eq:yy}
\end{eqnarray}
while by projecting on $\z_0\z_1$ we get
\begin{eqnarray}
\nu(\cL_{zz,zz}&+&\cL_{\1 z,\1 z}+\cL_{z\1,z\1})+\cL_{zz,xx}+\mu\cL_{zz,yy}+2h_x(\cL_{zz,\1 x}+\cL_{zz,x\1})+\nonumber \\
&+&2h_y(\cL_{zz,\1 y}+\cL_{zz,y\1})+2h_z(\cL_{zz,\1 z}+\cL_{zz,z\1})+h_z(\cL_{\1 z,\1\1}+\cL_{z\1,\1\1}) =0.
\label{eq:zz}
\end{eqnarray}
We will also need the three equations obtained by demanding that 1-local terms in Eq.~(\ref{eq:sum00}) should sum to zero. According to Lemma~\ref{lem:sum0} this for instance means that the coefficients in front of $\x_j\1_{j+1}$ and $\1_j\x_{j+1}$ should sum to zero. For $\x_j$ we therefore get the equation
\begin{eqnarray}
\cL_{\1 x,xx}&+&\cL_{x\1,xx}+\mu(\cL_{\1 x,yy}+\cL_{x\1,yy})+\nu(\cL_{\1 x,zz}+\cL_{x\1,zz})+2h_x(\cL_{\1 x,\1 x}+\cL_{x\1,x\1}+\cL_{\1 x,x\1}\cL_{x\1,\1 x})+\nonumber \\
&+&2h_y(\cL_{\1 x,\1 y}+\cL_{x\1,y\1}+\cL_{\1 x,y\1}+\cL_{x\1,\1 y})+2h_z(\cL_{\1 x,\1 z}+\cL_{x\1,z\1}+\cL_{\1 x,z\1}+\cL_{x\1,\1 z})=0.
\label{eq:x}
\end{eqnarray}
Similarly, for $\y_j$ we get,
\begin{eqnarray}
\mu(\cL_{\1 y,yy}&+&\cL_{y\1,yy})+\cL_{\1 y,xx}+\cL_{y\1,xx}+\nu(\cL_{\1 y,zz}+\cL_{y\1,zz})+2h_x(\cL_{\1 y,\1 x}+\cL_{y\1,x\1}+\cL_{\1 y,x\1}\cL_{y\1,\1 x})+\nonumber \\
&+&2h_y(\cL_{\1 y,\1 y}+\cL_{y\1,y\1}+\cL_{\1 y,y\1}+\cL_{y\1,\1 y})+2h_z(\cL_{\1 y,\1 z}+\cL_{y\1,z\1}+\cL_{\1 y,z\1}+\cL_{y\1,\1 z})=0,
\label{eq:y}
\end{eqnarray}
while for $\z_j$ we have
\begin{eqnarray}
\nu(\cL_{\1 z,zz}&+&\cL_{z\1,zz})+\cL_{\1 z,xx}+\cL_{z\1,xx}+\mu(\cL_{\1 z,yy}+\cL_{z\1,yy})+2h_x(\cL_{\1 z,\1 x}+\cL_{z\1,x\1}+\cL_{\1 z,x\1}+\cL_{z\1,\1 x})+\nonumber \\
&+&2h_y(\cL_{\1 z,\1 y}+\cL_{z\1,y\1}+\cL_{\1 z,y\1}+\cL_{z\1,\1 y})+2h_z(\cL_{\1 z,\1 z}+\cL_{z\1,z\1}+\cL_{\1 z,z\1}+\cL_{z\1,\1 z})=0.
\label{eq:z}
\end{eqnarray}
Matrix elements of $\cL_j$ are linear functions of unitary coefficients $d_{\bm{j}}$ and quadratic in dissipative Lindblad coefficients $c_{\bm{j}}$, the same also holds for all 6 equations above.

We proceed with the step (ii). First, we shall make an appropriate linear combination of previous 6 equations to remove the dependence on unitary coefficients $d_{\bm{j}}$, ending up with a single equation that is quadratic in dissipative $c_{\bm{j}}$, $c_i^* \tilde{C}_{i,j} c_j=0$. The matrix $\tilde{C}$ obtained by this procedure will not yet be negative-definite. Therefore, as a second step we shall combine it with certain equations obtained from the unitality condition (\ref{eq:unital}). To get rid of the coefficients of the Hamiltonian we observe that $d_j$ occur only in (some) off-diagonal elements of $\cL_j$. We can get all 
such coefficients to cancel each other by making a sum $E(\x\x)+\mu E(\y\y)+\nu E(\z\z)+2h_xE(\x)+2h_yE(\y)+2h_zE(\z)$, where by $E(\x\x)$ we denote Eq.~(\ref{eq:xx}) for $\x_j\x_{j+1}$, by $E(\x)$ Eq.~\ref{eq:x}, and similarly for the other 4. Packing 15 complex coefficients $c_{j}$ in a vector $\bm{c}$, the resulting equation is a quadratic form in $c_j$ only, $\bm{c}^\dagger \tilde{C} \bm{c}=0$, with a $15\times 15$ Hermitian matrix $\tilde{C}$ that depends on $\mu,\nu,h_x,h_y,h_z$ and is not yet negative-definite. We do not yet write it out. It turns out that we can get rid of all complex matrix elements of $\tilde{C}$, and at 
the same time make it also negative-definite, by combining it with quadratic forms obtained from the condition on unitality (\ref{eq:unital}). In particular, for all primitive 2-local operators, like $\x_j\x_{j+1}$, the unitality condition gives $\cL_{xx,\1\1}=0$, while for 1-site terms, like $\1_j\x_{j+1}$, we have $\cL_{\1 x,\1\1}+\cL_{x\1,\1\1}=0$. Note that the unitary part with $d_j$ is always unital and therefore these coefficients are absent from all $\cL_{nm,\1\1}$. Quadratic forms that we need are in fact not all $9+3$ equations ($9$ for 2-site operators and $3$ for 1-site), but instead only $9$ symmetric combinations
\begin{equation}
\begin{array}{rclrclrcl}
\cL_{xx,\1\1}\equiv \bm{c}^\dagger \tilde{C}^{\rm un}_{xx} \bm{c} & = & 0,& \cL_{yy,\1\1}\equiv \bm{c}^\dagger \tilde{C}^{\rm un}_{yy} \bm{c} & = & 0,& \cL_{zz,\1\1}\equiv \bm{c}^\dagger \tilde{C}^{\rm un}_{zz} \bm{c} & = & 0, \\
\cL_{zx,\1\1}+\cL_{xz,\1\1}\equiv \bm{c}^\dagger \tilde{C}^{\rm un}_{xz} \bm{c} & = & 0,& \cL_{yz,\1\1}+\cL_{zy,\1\1}\equiv \bm{c}^\dagger \tilde{C}^{\rm un}_{yz} \bm{c} & = &  0,& \cL_{xy,\1\1}+\cL_{yx,\1\1}\equiv \bm{c}^\dagger \tilde{C}^{\rm un}_{xy} \bm{c} & = & 0, \\
\cL_{\1 x,\1\1}+\cL_{x\1,\1\1}\equiv \bm{c}^\dagger \tilde{C}^{\rm un}_{x} \bm{c}&=&0,& \cL_{\1 y,\1\1}+\cL_{y\1,\1\1}\equiv \bm{c}^\dagger \tilde{C}^{\rm un}_{y} \bm{c}&=&0,& \cL_{\1 z,\1\1}+\cL_{z\1,\1\1}\equiv \bm{c}^\dagger \tilde{C}^{\rm un}_{z} \bm{c}&=&0,
\end{array}
\end{equation}
where we defined the corresponding matrices $\tilde{C}^{\rm un}$ of quadratic forms. Out of $\tilde{C}$ and $\tilde{C}^{\rm un}$ (which are purely imaginary) we can make a negative real symmetric matrix $C$,
\begin{eqnarray}
C\equiv\tilde{C}&-&(4h_x^2-\mu\nu)\tilde{C}^{\rm un}_{xx}-(4h_y^2-\nu)\tilde{C}^{\rm un}_{yy}-(4h_z^2-\mu)\tilde{C}^{\rm un}_{zz}-3h_x \tilde{C}^{\rm un}_x-3h_y\mu\tilde{C}^{\rm un}_y-3h_z \tilde{C}^{\rm un}_z-\nonumber \\
&-&4h_xh_y\tilde{C}^{\rm un}_{xy}-4h_xh_z\tilde{C}^{\rm un}_{xz}-4h_yh_z\tilde{C}^{\rm un}_{yz}.
\end{eqnarray}
Choosing an un-normalized basis of $C$ as $9$ symmetric combinations $\{\x_0\x_1,\y_0\y_1,\z_0\z_1\}$, followed by $\{\1_0\x_1+\x_0\1_1,\1_0\y_1+\y_0\1_1,\1_0\z_1+\z_0\1_1\}$ and $\{\z_0\y_1+\y_0\z_1,\z_0\x_1+\x_0\z_1,\y_0\x_1+\x_0\y_1\}$, and $6$ antisymmetric ones $\{\1_0\x_1-\x_0\1_1,\1_0\y_1-\y_0\1_1,\1_0\z_1-\z_0\1_1\}$ and $\{\z_0\y_1-\y_0\z_1,\x_0\z_1-\z_0\x_1,\y_0\x_1-\x_0\y_1\}$, $C$ has a block structure,
\begin{equation}
C=32
\left(
\begin{array}{cc}
C_1 & 0\\
0 & C_2
\end{array}\right),
\qquad
C_2=\left(\begin{array}{cc}
A & B^{\rm T} \\
B & \tilde{A}
\end{array}\right),
\quad
C_1=\left(\begin{array}{ccc}
A_1 & B_1 & B_2\\
B_1^{\rm T} & A_2 & B_3\\
B_2^{\rm T} & B_3^{\rm T} & A_3
\end{array}\right),
\label{eq:C}
\end{equation}
where $C_2$ is a $6\times 6$, $C_1$ a $9\times 9$ matrix, while all $A$ and $B$ are $3\times 3$ blocks,
\begin{eqnarray}
A&=&2\,\bm{h}\otimes\bm{h}-{\rm diag}(2h^2-\mu^2-\nu^2-\mu\nu,2h^2-1-\nu^2-\nu,2h^2-1-\mu^2-\mu), \nonumber \\
\tilde{A}&=&A-{\rm diag}(1,\mu^2,\nu^2),
\label{eq:A}
\end{eqnarray}
where $\bm{h}=(h_x,h_y,h_z)$ and $h=|\bm{h}|$, and 
\begin{equation}
B=\left( \begin{array}{ccc}
0 & -h_z(1+\mu+2\nu) & h_y(1+2\mu+\nu) \\
h_z(1+\mu+2\nu) & 0 & -h_x(2+\mu+\nu) \\
-h_z(1+2\mu+\nu) & h_x(2+\mu+\nu) & 0 
\end{array} \right),
\end{equation}
\begin{equation}
A_1=\left( \begin{array}{ccc}
-4h_y^2-4h_z^2-\mu^2-\nu^2& 4h_z^2 & 4 h_y^2 \\
4 h_z^2 & -4h_x^2-4h_z^2-1-\nu^2 & 4 h_x^2 \\
4 h_y^2 & 4 h_x^2 & -4h_x^2-4h_y^2-1-\mu^2 
\end{array} \right),
\end{equation}
\begin{equation}
A_2=A+{\rm diag}(2\mu\nu,2\nu,2\mu),
\end{equation}
\begin{equation}
A_3=A_2+4\, \bm{h}\otimes\bm{h}-{\rm diag}(1+12h_x^2,\mu^2+12h_y^2,\nu^2+12h_z^2),
\end{equation}
\begin{equation}
B_1=\left( \begin{array}{ccc}
0 & -\sqrt{2}h_y(\nu-1) & -\sqrt{2}h_z(\mu-1) \\
\sqrt{2}h_x(\mu-\nu) & 0 & \sqrt{2}h_z(\mu-1) \\
\sqrt{2}h_x(\nu-\mu) & \sqrt{2}h_y(\nu-1) & 0
\end{array} \right),
\end{equation}
\begin{equation}
B_2=\left( \begin{array}{ccc}
-\sqrt{32}h_yh_z & \sqrt{8}h_xh_z & \sqrt{8}h_xh_y \\
\sqrt{8}h_yh_z & -\sqrt{32}h_xh_z & \sqrt{8}h_xh_y \\
\sqrt{8}h_yh_z & \sqrt{8}h_xh_z & -\sqrt{32}h_xh_y
\end{array} \right),
\end{equation}
\begin{equation}
B_3=\left( \begin{array}{ccc}
0 & h_z(1+\mu-2\nu) & h_y(1+\nu-2\mu) \\
h_z(1+\mu-2\nu) & 0 & h_x(\mu+\nu-2) \\
h_y(1+\nu-2\mu) & h_x(\mu+\nu-2) & 0
\end{array} \right).
\end{equation}
We therefore obtained a quadratic form
 \begin{equation}
\bm{c}^\dagger C \bm{c}=0.
\label{eq:C0}
\end{equation}
The claim now is that $C$ (\ref{eq:C}) is negative-definite (i.e., all eigenvalues are negative), except for $\mu=\nu=h_y=h_z=0$, when it is negative semi-definite. 

Let us start by showing that $C_2$ is negative-definite except at $\mu=\nu=h_y=h_z=0$, i.e., that all eigenvalues are negative except at $\mu=\nu=h_y=h_z=0$, 
when some become $0$. One can show that either directly on $C_2$, by, for instance, evaluating determinants of the leading minors, seeing that they alternate in sign (Sylvester's criterion). This demands calculating the determinant of a $6\times 6$ matrix. Slightly less demanding procedure is by upper-bounding $C_2$ by another negative-definite matrix $C_2'$. Provided $C_2'$ has more structure 
(i.e., is more symmetric), proving its negative-definiteness could be simpler, while, because of $C_2' \ge C_2$, we would also know that $C_2$ can be possibly negative-semidefinite (some eigenvalues become zero) only at parameter values 
at which $C_2'$ becomes negative-semidefinite. By looking at the form of $C_2$ (\ref{eq:A}) we note that $\tilde{A}$ is almost equal to $A$. This leads us to make $C_2'=C_2+{\rm diag}(\mu\nu,\nu,\mu,1+\mu\nu,\mu^2+\nu,\nu^2+\mu)$. For $\mu,\nu \in [0,1]$ the diagonal elements that we added are non-negative and therefore $C_2' \ge C_2$ (i.e., their difference is a positive-semidefinite matrix). 
To show that $C_2'$ is negative-definite, except at special parameter values, we have to show that all eigenvalues are negative. Because in $C_2'$ the two diagonal $3\times 3$ blocks are now equal, while the off-diagonal block $B$ is an antisymmetric matrix, each eigenvalue of $C_2'$ is at least doubly degenerate. This can be seen by noting that if $(c_1,c_2,c_3,c_4,c_5,c_6)$ is an eigenvalue of $C_2'$, so is $(c_4,c_5,c_6,-c_1,-c_2,-c_3)$.  Because of the double degenerate 
eigenvalues the characteristic polynomial of $C_2'$ is effectively of degree $3$ instead of general $6$, being equal to $\lambda^3+b_2 \lambda^2+b_1\lambda+b_0$, with $b_2=4(1+\mu^2+\nu^2+2(h_x^2+h_y^2+h_z^2))$, while $b_1$ and $b_0$ are more complicated 
polynomials in parameters and we do not write them out. By Descartes' rule of signs we know that all roots are negative iff all coefficients $b_{0,1,2}$ are positive. For $b_2$ we readily see that $b_2 \ge 4$. Similarly one can show that $b_1 \ge 4$, while $b_0 \ge 0$, with the zero being attained only at $\mu=\nu=h_y=h_z=0$. $C_2'$ is therefore negative-definite, except at $\mu=\nu=h_y=h_z=0$, 
where it is negative-semidefinite. $C_2$ is therefore also negative-definite, except possibly at $\mu=\nu=h_y=h_z=0$. Calculating explicitly eigenvalues of $C_2$ 
at $\mu=\nu=h_y=h_z=0$ we see that indeed one is zero while the other $5$ are negative. $C_2$ is therefore negative-definite except at $\mu=\nu=h_y=h_z=0$ when it is negative-semidefinite.

Let us now show that $C_1$ is also negative-semidefinite, except at $\mu=\nu=h_y=h_z=0$. Because of its larger size things are more involved than with $C_2$ (using computer algebra is advisable), but the general idea is the same. We first construct a simpler upper-bound matrix $C_1'=C_1+{\rm diag}(\mu^2+\nu^2,1+\nu^2,1+\mu^2,0,0,0)$, $C_1' \ge C_1$. $C_1'$ is simpler because we can guess one 
eigenvector $\bm{v}=(1,1,1,0,0,0,0,0,0)$ with the corresponding eigenvalue $\lambda=0$. Let us denote the $8$-dimensional subspace orthogonal to $\bm{v}$ by ${\cal H}_\perp$, and the total $9$ dimensional space by ${\cal H}$. Any vector $\psi \in {\cal H}$ can be written as $\psi=\alpha {\bm v}+\psi_\perp$, where 
$\psi_\perp \in {\cal H}_\perp$. We have $ \psi^\dagger C_1 \psi \le \psi^\dagger C_1' \psi=\psi_\perp^\dagger C_1' \psi_\perp$. The claim now is that $C_1'$ is 
negative-definite on ${\cal H}_\perp$, meaning that $\psi_\perp^\dagger C_1' \psi_\perp <0$, except at special points. This means that, if $\psi_\perp \neq 0$ one also has $\psi^\dagger C_1 \psi <0$. In addition, one can easily check that $\bm{v}^\dagger C_1 \bm{v}=-2(1+\mu^2+\nu^2)<0$, so that $C_1$ is negative-definite, except 
possibly at special points at which $C_1'$ is not negative-definite. We therefore proceed by considering $C_1'$ on ${\cal H}_\perp$, showing that it is negative-definite, except at special points. $C_1'$ defined on ${\cal H}_\perp$ is an $8\times 8$ matrix $\tilde{C}_1'$. We shall show that $\tilde{C}_1'$ is 
negative-definite by using Sylvester's criterion, stating that a matrix is negative-definite iff the determinants of $[\tilde{C}_1']_j, j=1,\ldots,8$, alternate in sign, and $\det{[\tilde{C}_1']_1}<0$, where $[\tilde{C}_1']_j$ is a matrix obtained by 
taking the first $j$ rows and columns of $\tilde{C}_1'$. First, we observe that all $8$ determinants are polynomials in $\bm{h}$ with only even powers of $h_{x,y,z}$ being present, and that the coefficients in front of each monomial in 
magnetic fields is a polynomial in $\mu$ and $\nu$. One can show that all these coefficients are non-positive for $\det{[\tilde{C}_1']_j}$ and odd $j$, and non-negative for $\det{[\tilde{C}_1']_j}$ and even $j$. Therefore, except at points at which some of these coefficients become zero, also the whole determinant has negative/positive sign and $\tilde{C}_1'$ is negative-definite. Special point for 
which some coefficients are zero is always $\mu=\nu=0$. This is the only possibility at which $\tilde{C}_1'$ can become negative-semidefinite. Note that so far we did not care about $h_{x,y,z}$ at which the determinants are zero because 
at the end we will anyway have to check the negativity of $C_1$. For $\mu=\nu=0$ the matrix $C_1$ is sufficiently simple so that we can check its eigenvalues 
directly. Calculating the characteristic polynomial we can see that all coefficients are strictly positive numbers, except the two coefficients in front of $\lambda^0$ and $\lambda$, that are both $0$ for $h_y=h_z=0$. This shows that $C_1$ is negative-definite, except for $\mu=\nu=h_y=h_z=0$ when it is negative-semidefinite (two eigenvalues are zero).

We have proved that $C$ is negative-definite, except for $\mu=\nu=h_y=h_z=0$. Apart from that special point Eq.~(\ref{eq:C0}) therefore has no non-zero solutions for $\bm{c}$ and, correspondingly, a 2-site dissipator with a single Lindblad operator that would globally conserve such $a$ does not exist.

Step (iii), proving that the same conclusion holds also for a Liovillian with more than one Lindblad operator, is now easy. Let us write these Lindblad operators as $L^{(k)}=\sum_{\bm{j}} c^{(k)}_{\bm{j}} \sigma_0^{j_1}\sigma_1^{j_2}$. Due to linearity we would, instead of Eq.~(\ref{eq:C0}), get equation $\sum_k \bm{c}^{(k)\dagger} C \bm{c}^{(k)}=0$, where $C$ is the same as for the case of a single $L$. Due to negative-definiteness of $C$ there would again be no solutions. 

\section{Sketch of the proof for a $3$-local Liouvillian}
\label{AppB}

We obtain a real symmetric matrix $C$ by first forming $\tilde{C}$ from equations $E(\x\x)+\mu E(\y\y)+\nu E(\z\z)+h_x E(\x)+h_y E(\y)+ h_z E(\z)$, and 
then $C=\tilde{C}+(\mu\nu-h_x^2)\tilde{C}^{\rm un}_{xx}+(\nu-h_y^2)\tilde{C}^{\rm un}_{yy}+(\mu-h_z^2)\tilde{C}^{\rm un}_{zz}-2h_x\tilde{C}^{\rm un}_{x}-2\mu h_y\tilde{C}^{\rm un}_{y}-2\nu h_z\tilde{C}^{\rm un}_{z}-h_xh_y\tilde{C}^{\rm un}_{xy}-h_xh_z \tilde{C}^{\rm un}_{xz}-h_yh_z\tilde{C}^{\rm un}_{yz}-(1+h_x^2)\tilde{C}^{\rm un}_{x\1 x}-(\mu^2+h_y^2)\tilde{C}^{\rm un}_{y\1 y}-(\nu^2+h_z^2)\tilde{C}^{\rm un}_{z\1 z}-h_xh_y\tilde{C}^{\rm un}_{x\1 y}-h_xh_z \tilde{C}^{\rm un}_{x\1 z}-h_yh_z\tilde{C}^{\rm un}_{y\1 z}-2h_x \tilde{C}^{\rm un}_{xxx}-2\mu h_y\tilde{C}^{\rm un}_{yyy}-2\nu h_z\tilde{C}^{\rm un}_{zzz}-\mu h_x \tilde{C}^{\rm un}_{xyy}-\nu h_y \tilde{C}^{\rm un}_{yzz}-\nu h_x \tilde{C}^{\rm un}_{xzz}-\mu h_z \tilde{C}^{\rm un}_{zyy}-h_y \tilde{C}^{\rm un}_{yxx}-h_z \tilde{C}^{\rm un}_{zxx}+h_y \mu \tilde{C}^{\rm un}_{\1\1 x}+h_z\nu\tilde{C}^{\rm un}_{\1\1 z}+h_x \tilde{C}^{\rm un}_{\1\1 x}$, where $\tilde{C}^{\rm un}$ are 
equations obtained from the unitality condition for $3$-local Liouvillians (\ref{eq:unitality3}). Because of a large size of $C$ it is difficult to analytically prove that it is negative-definite for all values of parameters. One can analytically check though that this is indeed the case for specific values of parameters or when some parameters are zero and $C$ simplifies: we have analytically checked that such $C$ is negative-definite for the XYZ type of $a_j=\x_j\x_{j+1}+\mu \y_j\y_{j+1}+\nu \z_j\z_{j+1}$ ($\mu,\nu \in [0,1]$, apart from $\mu=\nu=0$) by using Descartes' rule on a characteristic polynomial (due to symmetries and degeneracies one can reduce the eigenvalue problem to polynomials of smaller degree: four different polynomials are of degree 6, eleven different are of degree 3, while 6 eigenvalues can be explicitly expressed); similarly, one can analytically check that $C < 0$ for the Ising like $a_j=\x_j\x_{j+1}+h_y \y_{j+1}+h_z \z_{j+1}$ ($h_y,h_z \in \mathbbm{R} \setminus
  \{0\}$) (16 explicit eigenvalues, three different polynomials of degree 3, each occurring three times, two of degree 3, each two times, two different of degree 3, two different of degree 6, and two different of degree 7); 
for the XXZ Heisenberg type $a_j=\x_j\x_{j+1}+\y_j\y_{j+1}+\nu \z_j\z_{j+1}+h_z \z_{j+1}$ ($h_z \in \mathbbm{R}$, $\nu \in [0,1]$) (19 explicit eigenvalues, two different polynomials of degree 4, two same polynomials of degree 4, two same of degree 6, two same of degree 8); for the XX model in a transverse field, $a_j=\x_j\x_{j+1}+\y_j\y_{j+1}+h_x \x_{j+1}$ (4 explicit eigenvalues, one polynomial of degree 4, one of degree 5, three different of degree 6, two different of degree 10, and one of degree 12).
For other generic values of parameters we have scanned parameter ranges and numerically checked that $C$ is negative-definite, except at the mentioned parameters. Therefore, using a 3-site Liouvillian $\cL_j$ one can globally conserve only those 2-site $a_j$ that can already be (locally) conserved by a 2-site $\cL_j$.


\section*{Acknowledgements}
G.B. and G.C. acknowledge the support by 
MIUR-PRIN project \emph{Collective quantum phenomena: 
From strongly correlated systems to quantum simulators} 
and by Regione Lombardia. M.\v Z. would like to thank Universit\`a degli Studi dell'Insubria for hospitality during his stay there.

\end{document}